\documentclass[10pt,journal,a4paper]{IEEEtran}
\usepackage{amsmath}
\usepackage{amsthm}
\usepackage{amsfonts}
\usepackage{amssymb}
\usepackage{algpseudocode}
\usepackage{subfigure}
\usepackage{algorithm}
\usepackage{xspace}
\usepackage{color}
\usepackage{enumerate}
\usepackage{paralist}
\usepackage{flushend}
\usepackage{cite}
\usepackage{graphicx}
\usepackage{epstopdf}

\newtheorem{lemma}{Lemma}
\newtheorem{corollary}{Corollary}
\newtheorem{theorem}{Theorem}

\usepackage{graphicx}
\newcommand{\gt}{GT criterion\xspace}
\renewcommand{\eqref}[1]{Eq. \ref{#1}}
\begin{document}
	\title{Strategies for Utility Maximization in Social Groups with Preferential Exploration}
	\author{
		\IEEEauthorblockN{Saurabh Aggarwal, Joy Kuri}\\
		\IEEEauthorblockA{Department of Electronic Systems Engineering,\\ Indian Institute of Science, Bangalore, India,\\
			Email: saggarwal@dese.iisc.ernet.in, kuri@dese.iisc.ernet.in}
	}
	\maketitle
	\bibliographystyle{IEEEtran}
	\begin{abstract}
We consider a ``Social Group'' of networked nodes, seeking a ``universe'' of segments for maximization of their utility. Each node has a subset of the universe, and access to an expensive link for downloading data. Nodes can also acquire the universe by exchanging copies of segments among themselves, at low cost, using inter-node links. While exchanges over inter-node links ensure minimum or negligible cost, some nodes in the group try to exploit the system. We term such nodes as `non-reciprocating nodes' and prohibit such behavior by proposing the ``Give-and-Take'' criterion, where exchange is allowed iff each participating node has segments unavailable with the other. Following this criterion for inter-node links, each node wants to maximize its utility, which depends on the node's segment set available with the node.

Link activation among nodes requires mutual consent of participating nodes. Each node tries to find a pairing partner by preferentially exploring nodes for link formation and unpaired nodes choose to download a segment using the expensive link with segment aggressive probability. We present various linear complexity decentralized algorithms based on \emph{Stable Roommates Problem} that can be used by nodes (as per their behavioral nature) for choosing the best strategy based on available information. Then, we present decentralized randomized algorithm that performs close to optimal for large number of nodes. We define \emph{Price of Choices} for benchmarking performance for social groups (consisting of non-aggressive nodes only). We evaluate performances of various algorithms and characterize the behavioral regime that will yield best results for node and social group, spending the minimal on expensive link. We consider social group consisting of non-aggressive nodes and benchmark performances of proposed algorithms with the optimal. We find that \emph{Link For Sure} Algorithm performs very close to optimal and subsidizes the need to consider long term strategies for maximization of node's utility.
\end{abstract}
	\section{Introduction and Related work}
\label{sec:intro}
We consider a set of devices using Cellular Network assisted Device-to-Device (CN-D2D) communication \cite{design_network_d2d,d2d_socialgroup} for fulfilling their content needs (such as a high definition media, set of files for software update,  etc). In CN-D2D, devices connect to base station using the Long-Term Evolution (LTE), whereas devices communicate among themselves using WiFi Direct. Content is broken into smaller segments to facilitate content dissemination among devices. Each device downloads the initial subset of segments directly from the base station (using LTE). Devices store the data and disseminate it to other devices in a Device-to-Device (D2D) group \cite{d2d_storage}. Cost of exchanging segments among devices is negligible as compared to the cost of downloading segments from the cellular network. In CN-D2D, the base station facilitates inter-device communication and also acts as content provider on request of devices.

CN-D2D has a network architecture, which is quite similar to being observed in many other networks such as hybrid Peer-to-Peer Content Distribution Network \cite{anatomy,cdnstatus}, Cloud assisted Peer-to-Peer Network \cite{cloud_p2p}, and Direct Connect Networks (popularly known as DC++) \cite{dc++,p2pdc++}. We consider a scenario in which self-interested nodes/devices/peers/users are looking for some common content. Such a network architecture consisting of inter-connected nodes interested in obtaining common content is termed as a \emph{Social Group} \cite{social_group}. Nodes in Social Group are unknown to each other and have lack of trust among themselves. The desire to obtain the same content at minimal cost is the only shared characteristic among the group members. To facilitate exchange among nodes, we assume there is entity called \emph{facilitator} in network, whose role is to gather information about networked nodes. Tracker in Cloud assisted Peer-to-Peer Network\cite{cloud_p2p}, base station in CN-D2D \cite{design_network_d2d}, central hub/server in Direct Connect Networks \cite{dc++,p2pdc++}, content distributer in Peer-to-Peer Content Distribution Network \cite{cdnstatus,anatomy} can act as facilitator.

\begin{figure}[h]
	\centering
	\includegraphics[width=\linewidth]{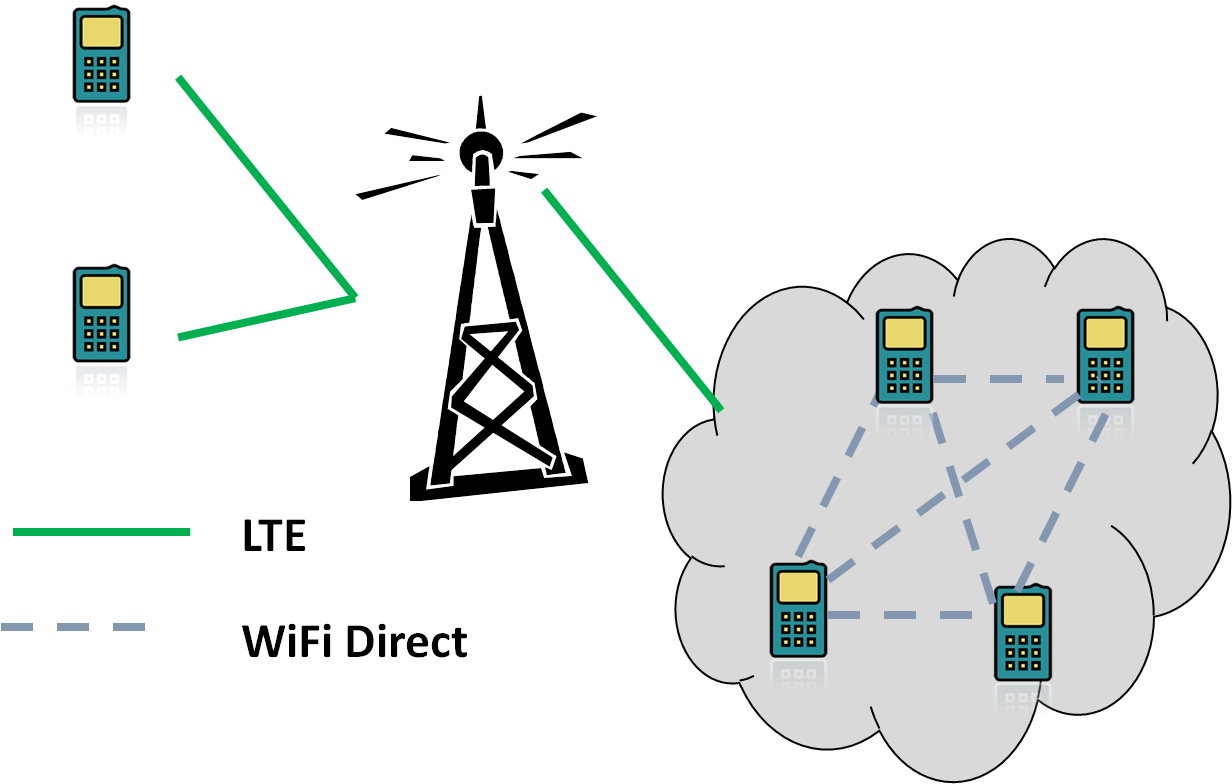}
	\caption{Base Station and group of networked devices using Cellular Network Device-to-Device communication}
	\label{fig:social_group}
\end{figure}
In this paper, any node $A$'s ``selfishness" translates to a refusal to provide copies of segments to another node $B$, unless $A$ receives something in return from $B$. We refer to this as the \emph{Give-and-Take Criterion} \cite{agg_sadhana,social_optimum,vaze_ncc}. After exchange, both $A$ and $B$ will have the \emph{union} of their individual pre-exchange segment sets. Each node wants to acquire the \emph{universe} of segments at low cost via predominantly local exchanges among nodes, subject to the \gt, otherwise from the content provider in social group.

Cooperation among nodes and reduction of download costs in social groups \cite{social_similarity} in ensured by the \gt. Presence of the \gt ensures the absence of \emph{non-reciprocating nodes}, i.e. nodes keen to download segments without uploading any segments in return. Non-reciprocating behavior of nodes matches free-riding behavior observed in various societal systems and communication networks\cite{free_rider,free_economy}. For discouraging such non-reciprocating behavior, various methods and polices have been proposed and implemented in various societal systems and communication networks \cite{freewhite,overfreeride,cheapfreeriding,effort_based,global_fair,free_riding,g2g_free,free_economy,free_public}. Policy counterparts of these can be defined for social groups, but nodes can still exhibit non-reciprocating behavior by doing collusion, identity fraud etc. \cite{sybil,free_riding,sybil_p2p}. For instance, peers practicing whitewashing and sybil attack\cite{sybil_p2p} in Peer-to-Peer network following Bit-torrent's Tit-for-Tat mechanism are able to free ride. However, \gt is immune to such malpractices as it requires node to upload segments to the other node \emph{immediately}. Hence, to participate in GT-compliant local
exchanges, each node needs to have \emph{some} segment(s) with itself --- otherwise, it cannot provide any segments to other nodes as  the \gt fails. For this reason, each node downloads an initial set of segments that allows him to participate in local exchanges in compliance with \gt.

The \gt applies to local exchanges of segment-copies among nodes, and such exchanges are very cheap. Also, each node can download segments from the expensive link (example LTE link in CN-D2D systems); \gt does not apply for such downloads. As noted above, each node downloads an \emph{initial subset} of segments using the expensive link for kick-starting local exchanges. 

Authors in \cite{social_optimum} have studied the problem of centralized scheduling of GT-compliant inter-node exchanges for maximizing the aggregate cardinality. The motivation was to reduce aggregate cost of download over the expensive link. However, the assumption of centralized scheduling need not be true, and we extend the work where nodes act in a decentralized manner on their own for maximization of their own utility functions. 

We consider a decentralized version of the problem being studied in \cite{social_optimum}. Unlike \cite{social_optimum}, nodes can accept or reject a link for activation. For activating a link, mutual consent from the participating nodes is necessary. Nodes will choose the pairing nodes based on the utility that will be derived by them in case of link activation. 

Nodes might choose to ignore nodes which will not return significant utility. To factor such preferential behavior of nodes, we introduce \emph{Preferential Exploration Factor (PEF)} in system model. In some instances, requirement of mutual consent might leave \emph{some}\footnote{Lemma \ref{lem:link_exists1} shows that \emph{at least} one pair of nodes will be give mutual consent for activation of link.} nodes unpaired. These unpaired nodes might choose to download new segment via expensive link based on their aggressiveness for new segments. Segment Aggressiveness for new segments is factored into system by introducing \emph{Segment Aggressive Probability(SAP)} in the system. We propose practically implementable decentralized algorithm for nodes, to determine the best strategy for the node. 

For special cases, where nodes might preferentially explore avenues for link formation, but does not downloads new segments (after the downloading of initial segment set), we propose practically implementable decentralized algorithm for determining the best strategy for node. We define \emph{Price of Choices} to compare the equilibrium scenarios emerging due to decentralized selection of strategies by nodes with the optimal aggregate utility, which can be obtained through centralized scheduling. Also, we propose algorithm which can work in absence of facilitator, but still performing close to optimal.

This paper makes the following contributions:
\begin{itemize}
	\item We use the \gt, to prohibit \emph{non-reciprocating behavior} in social groups. This, we believe, will help in understanding the fundamental principles of data sharing networks consisting of only selfish but contributing nodes.
	\item We propose a practically implementable decentralized algorithms for nodes, with preferential nature of exploration (factored into \emph{Preference Exploration Factor}) and aggression for new nodes (factored into \emph{Segment Aggressive Probability}) namely, \emph{Limited Stable Pairing Algorithm} for deciding upon its strategy in each slot.
	\item For scenarios consisting of nodes which do not download new segments (after downloading new segment set), we use \emph{Price of Choices} to benchmark performance of our proposed decentralized algorithms with the optimal.
	\item We propose \emph{Decentralized Randomized algorithm} which can work in the absence of \emph{facilitator}, but still performing close to optimal.
\end{itemize}

Our paper will be arranged as follows: Section \ref{sec:sysmodel} will describe the system model in detail. Section \ref{sec:analysis} presents a detailed analysis and propose different algorithms depending on various scenarios. We define Price of Choices in Section\ref{sec:poc}. We evaluate and benchmark performances of various algorithms in Section \ref{sec:performance}. At last we will give our conclusions in Section \ref{sec:conclusions} followed by some open ended problems in Section \ref{sec:future}.
	\section{System Model and Problem Formulation}
\label{sec:sysmodel}
We consider a set of nodes, $\mathcal{M}=\left\{1,2,\cdots, m\right\}$ and a universe $\mathcal{N}=\left\{1,2,\cdots, n\right\}$ of segments. Each node $i\in\mathcal{M}$ has an initial collection of segments $O_i\subset\mathcal{N}$. 

\emph{Give-and-Take (GT) criterion:} Two nodes $i, \ j\in \mathcal{M}$ with segment sets $X_i$ and $X_j$, respectively, can exchange segments 
if and only if $X_i\cap X_j^c \neq\emptyset \quad \text{and} \quad X_i^c\cap X_j \neq\emptyset$, i.e. node $j$ has at least one segment which is unavailable with node $i$ and vice versa. 
After exchange, both nodes have the segment set $X_i\cup X_j$. 

We consider the set of nodes $\mathcal{M}$ as the vertices in an undirected graph $G$, where an edge or {\it link} exists between two vertices $i,j\in\mathcal{M}$ if and only if they satisfy \gt. We denote the link between nodes $i$ and $j$ by the unordered 2-tuple $(i,j)$. Both nodes $i$ and $j$ will posses their individual segment sets before exchange. Link $\left(i,j\right)$ can be activated with the consents of node $i$ and $j$ only.

We observe the system at epochs of various discrete decision slots, $r\in\mathbb{N}$.  Dynamic graph at the  beginning of $r^{th}$ decision slot is denoted by $G(r)=\left(\mathcal{M},\mathcal{L}(r)\right)$, where $\mathcal{L}(r)$ denote the set of links satisfying \gt. We define $\mathcal{L}_a(r)\subseteq \mathcal{L}(r)$ to be the set of links being activated in $r^{th}$ slot. Also, $l_i(r)$ denotes the set of the nodes, which satisfy the \gt with node $i$ in $r^{th}$ slot.

Let $O_i(r)$ be the segment set available with node $i$ in the beginning of $r^{th}$ slot. In $r^{th}$ slot, node $i$ will try to pair with one of the nodes with whom \gt is satisfied based on its preferences. If a node is unpaired, then unpaired node will download one new segment uniformly at random from server with a probability based on its aggressiveness for new segments. Strategy set for node $i$ in $r^{th}$ decision slot is given by, 
$$S_i(r)=l_i(r)\bigcup\{0\},$$ 
where $\{0\}$ denotes no pairing strategy of node $i$. If strategy $\{0\}$ is chosen by node $i$, node $i$ will download one new segment uniformly at random from the server with segment with \emph{Segment Aggressive Probability (SAP)}, $a_i(r)\in \left[0,1\right]$ (depending upon node's aggressiveness for new segments). Node $i$ might limit its preferences for exploring possibilities of link formation based on \emph{Preference Exploration Factor (PEF)}, $e_i(r)\in\left[0,1\right]$.

Strategy chosen by node $i$ in $r^{th}$ decision slot is denoted by $s_i(r)\in S_i(r)$. $s_{-i}(r)$ denotes set of strategies adopted by all nodes in $\mathcal{M}\backslash\left\{i\right\}$. Link $\left(i,j\right)$ can be activated in $r^{th}$ slot iff $s_i(r)=j$ and $s_j(r)=i$. $\mathbf{S}(r)=\left[s_i(r)\right]_{m\times 1}$ is a column vector, where $i^{th}$ element denotes the strategy chosen by node $i$ in $r^{th}$ slot. Also, $\mathbf{\hat{S}}(r)=\left[\mathbf{S}(1),\mathbf{S}(2)\cdots \mathbf{S}(r)\right]$ is a collection of all column vectors $\mathbf{S}(r)$. $\mathbf{C}_{exp}(r)=\left[c_i(r)\right]_{m\times 1}$ is a column vector, where $i^{th}$ element denotes the number of segments downloaded by node $i$ till the end of $r^{th}$ slot.

Node $i$'s utility\footnote{For the purpose of analysis, we have assumed $u_i(r)$ to be strictly increasing w.r.t. to cardinality of segment set as outlined in A\ref{asmp:strict_inc}., $u_i(r)$ in the beginning of $r^{th}$ decision slot, is a function of the segment set $O_i(r)$ available with node $i$, i.e. $u_i(r): O_i(r)\rightarrow \mathbb{R}^+$}. Also, $u_i(r+1)=u_i(r)+\Delta u_i(r,s_i(r),s_{-i}(r))$, where $\Delta u_i(r,s_i(r),s_{-i}(r)):(O_i(r),s_i(r),s_{-i}(r))\rightarrow \mathbb{R}^+\bigcup\left\{0\right\}$. $\bar{U}(r)=\left[u_1(r),u_2(r)\cdots u_m(r)\right]$ denotes the vector of utilities in the beginning of $r^{th}$ slot. We also define the utilitarian function, $u\left(r\right)=\sum_{i\in\mathcal{M}} u_i\left(r\right)$ to be used in Section \ref{sec:poc}. 

A series of activation of links in decision slots will eventually lead to a slot $r_{end}$, where \gt is not satisfied for any pair of nodes in $\mathcal{M}$. Also, $\bar{U}(r)=\bar{U}(r_{end})\, \forall\, r\geq r_e$.

Given an initial collection of segment sets $O_i(1)$, node $i$ chooses a series of strategies $\left[s_i(1),s_i(2)\cdots s_i(r_{end}-1)\right]$ so that utility $u_i(r_{end})$ is maximized.
\subsection*{Assumptions}
\label{subsec:assumption}
\begin{enumerate}[{A}1:]
	\item Cost of exchange among nodes is zero or negligible as compared to the cost of downloading the segment from outside the set $\mathcal{M}$ similar to \cite{social_similarity}. Cost of downloading any one segment using the expensive link is $c_{exp}$.
	\item $O_i(1)\neq\emptyset, \text{ and }O_i\neq \bigcup_{j\in\mathcal{M}} O_j=\mathcal{N}, \quad \forall\ i\in \mathcal{M}$.
	\item \label{asmp:strict_inc}Each of the node has the same utility function, $f\left(\cdot\right)$, which is a strictly increasing function of cardinality of segment set available with node, i.e.
	\begin{eqnarray*}
	u_i(r)=f\left(\left|O_i(r)\right|\right)\quad \forall i\in\mathcal{M}\text{ and }r\in\mathbb{N}\\
	\text{and }\forall x,y\in\mathbb{N},\quad x>y\implies f\left(x\right)>f\left(y\right)
	\end{eqnarray*}
	\item \label{asmp:node_knows} In the beginning of each slot, each node has information about the various segment sets available with all other nodes (provided by facilitator), though nodes choose their strategies in a decentralized manner. (We remove this assumption while analyzing \emph{Randomized Strategy Selection} in Section \ref{subsec:rand_sel}.)
\end{enumerate}
\subsection*{Maximizing the Node's Objective}
\label{subsec:objective}
In the problem, each node chooses a series of strategies for maximization its own utility. However, utility gain of choosing a strategy depends on choice of strategies by other nodes as well. Objective function from node's perspective can be written as,
\begin{align}\label{eq:probform}
&\max & u_i(r_{end})=f\left(\left|O_i(r_{end})\right|\right)& \\ \nonumber
&\textrm{subject to: } & s_i(r)\in S_i(r) \quad\forall& \begin{array}{c}
i\in\mathcal{M},\\
r\in\left\{1,2\cdots r_{end}-1\right\}
\end{array}
\end{align}

Each node tries to maximize its utility by means of choosing strategies in each slot so that utility of the node towards the end can be maximized. While, choosing a strategy in each slot, node $i$ should consider its impact on future strategies of the other nodes, which in turn will be affecting the future strategies of node $i$. For the sake of analysis in this paper, we consider nodes to be myopic in deciding their strategies in each slot. i.e., nodes explore various opportunities for link formation and choose the one returning the best immediate utility gain. Immediate utility gain for node $i$ in $r^{th}$ slot can be defined as,
\begin{eqnarray}
\Delta u_i(r,s_i(r),s_{-i}(r)=\begin{cases}
g\left(i,s_i(r),r\right), \text{if} \, s_{s_i(r)}(r)=i\\
0 \quad\quad\text{Otherwise}
\end{cases}\nonumber\\
\text{where, }g\left(i,j,r\right)=f\left(\left|O_j(r)\bigcup O_i(r)\right|\right)-f\left(\left|O_i(r)\right|\right) \nonumber
\end{eqnarray}

Our problem of finding an appropriate pairing nodes in each slot can be modeled as \emph{Stable Roommates Problem with Ties and Incomplete Lists (SRPTI)}\cite{srp2002}. Each node $i\in\mathcal{M}$ has a preference list\footnote{Node $i$ may choose to ignore some nodes in $l_i(r)$ based on value of PEF for node $i$ in $r^{th}$ slot,$e_i(r)$} for set of nodes in $l_i(r)$, such that, node $j_1$ is strictly preferred over node $j_2$ in $r^{th}$ slot, iff $g(i,j_1,r)>g(i,j_2,r)\, \forall\, j_1,j_2\in l_i(r)$. Stable Matching as defined in \cite{srp1985,srp2002} exists iff all nodes are able to find a pairing node. Stable matching need not exist always. But, there might exist some matchings in which only some nodes are able to find stable pairing node. Such matchings consisting of at least 1 stable pair are termed as \emph{Partial Stable Matching (PSM)}.

Imposition of restrictions on utility functions (as per A\ref{asmp:strict_inc}) in context of our problem, ensures existence of PSM in each slot. We obtain PSM in each slot using the concepts devised in Phase 1 of proposed algorithms of \cite{srp1985,srp2002}.

Hence, we solve (\ref{eq:probform}) by repeatedly modeling it as SRPTI. We obtain the utilities obtained in a decentralized manner under the various behavioral factors such as SAP and PEF. Analysis of the utilities towards the end, helps in characterizing the download costs incurred by the individual nodes in the social group over the expensive link for obtaining a desired common content.
	\section{Analysis}
\label{sec:analysis}
\subsection{Deterministic Strategy Selection}
We assume that the nodes are myopic and choose strategies based on immediate utility increment in a slot. Strategy being chosen by node $i$, will yield positive utility iff other node also choses node $i$ as the strategy. A pair of nodes $i$ and $j$ forms a \emph{stable pair}, iff $\not\exists\,k$, which $i$ prefers over $j$, and node $i$ is not preferred by node $k$ over its current partner. Pair of nodes returning the best immediate utility in a slot is a \emph{stable pair}, however, vice versa need not be true.

Best immediate utility to node $i$ in $r^{th}$ slot is given by $\underset{s_i(r),s_{-i}(r)}{\max} \Delta u_i\left(r,s_i\left(r\right),s_{-i}\left(r\right)\right)$. 

For each $r\leq r_e$, consider a directed graph based on first preferences of the nodes $G_f(r)=\left(\mathcal{M},\mathcal{L}_f(r)\right)$, where $\mathcal{L}_f(r)=\left\{i\rightarrow j: j\in\arg\underset{j'\in l_i(r)}{ \max} g(i,j',r)\right\}$. Link between node $i$ and $j$ is bi-directional $i\leftrightarrow j$ iff $i\rightarrow j$ and $j\rightarrow i$ are present in $\mathcal{L}_f(r)$. A path in $G_f(r)$ is a cycle, if there exists a sequence of nodes $\left(i_1,i_2\cdots i_q,i_1\right)$ such that  links $i_p\rightarrow i_{p+1}\in\mathcal{L}_f(r)\, \forall\, 1\leq p\leq q-1$ and $i_q\rightarrow i_1\in\mathcal{L}_f(r)$.
\begin{lemma}
	\label{lem:link_exists1}
	For all $r\leq r_{end}$, each cycle in $G_f(r)$ contains at least one bidirectional link.
\end{lemma}
\begin{proof}
	There are no self loops by definition of \gt. Hence, cycle will have more than 1 node. Also, by definition of bidirectional link, cycle with $2$ nodes is a bidirectional link.
	
	Let us consider a cycle with $q$ nodes $\left(2<q\leq m\right)$ namely $\left(i_1,i_2,\cdots i_q,i_1\right)$. If possible, let there be no bidirectional link in the cycle.
	
	Since, node $i_2$ returns the best utility to node $i_1$, therefore, 
	\begin{eqnarray}
		\label{eq:1st_inequality1}
		g(i_1,i_2,r)&\geq& g(i_1,i_q,r)\nonumber\\
		\Rightarrow f\left(\left|O_{i_2}(r)\bigcup O_{i_1}(r)\right|\right)&\geq& f\left(\left|O_{i_q}(r)\bigcup O_{i_1}(r)\right|\right)
	\end{eqnarray}
	Also, existence of equality in \eqref{eq:1st_inequality1} $\Rightarrow i_1\rightarrow i_q \in\mathcal{L}_f(r)$, therefore, link between node $i_1$ and $i_q$ is bidirectional. Rewriting \eqref{eq:1st_inequality1},
	\begin{eqnarray}
		\label{eq:1st_inequal}
		f\left(\left|O_{i_2}(r)\bigcup O_{i_1}(r)\right|\right) > f\left(\left|O_{i_q}(r)\bigcup O_{i_1}(r)\right|\right)
	\end{eqnarray}
	
	Similarly, node $i_3$ returns the best utility to node $i_2$, therefore,
	\begin{eqnarray}
		\label{eq:2nd_inequal}
		f\left(\left|O_{i_3}(r)\bigcup O_{i_2}(r)\right|\right) > f\left(\left|O_{i_1}(r)\bigcup O_{i_2}(r)\right|\right)
	\end{eqnarray}
	
	Continuing in this manner we will get a set of inequalities like \eqref{eq:1st_inequal} and \eqref{eq:2nd_inequal}
	%
	%
	

	\begin{eqnarray*}
		f\left(\left|O_{i_2}(r)\bigcup O_{i_1}(r)\right|\right)&>& f\left(\left|O_{i_q}(r)\bigcup O_{i_1}(r)\right|\right)\\
		f\left(\left|O_{i_3}(r)\bigcup O_{i_2}(r)\right|\right)&>& f\left(\left|O_{i_1}(r)\bigcup O_{i_2}(r)\right|\right)\\
		f\left(\left|O_{i_4}(r)\bigcup O_{i_3}(r)\right|\right)&>& f\left(\left|O_{i_2}(r)\bigcup O_{i_3}(r)\right|\right)\\
		&\vdots&\\
		f\left(\left|O_{i_q}(r)\bigcup O_{i_{q-1}}(r)\right|\right)&>& f\left(\left|O_{i_{q-2}}(r)\bigcup O_{i_{q-1}}(r)\right|\right)\\
		f\left(\left|O_{i_1}(r)\bigcup O_{i_q}(r)\right|\right)&>& f\left(\left|O_{i_{q-1}}(r)\bigcup O_{i_q}(r)\right|\right)\\
	\end{eqnarray*}
	Since the above set of inequalities is inconsistent, there is a contradiction. Therefore, there exists at least one bidirectional link in the cycle.
\end{proof}
\begin{corollary}[Existence of Partial Stable Matching]
	For all $r\leq r_{end}$, PSM exists in $r^{th}$ slot.
\end{corollary}
\begin{proof}
	Lemma 1 proves existence of a pair of nodes $i$ and $j$ such that activating link $\left(i,j\right)$ returns the best immediate utility to nodes $i$ and $j$ in $r^{th}$ slot. Such a pair of nodes is a \emph{Stable Pair}. Hence, Partial Stable Matching exists.
\end{proof}

\emph{Partial Stable Matching (PSM)} can also be defined as a collection of stable pairs such that no two stable pairs share a common node. PSM such that no two nodes among the unpaired nodes satisfy the \gt is defined as \emph{Maximal Partial Stable Matching (MPSM)}. Also, we define \emph{Limited Preference Stable Matching (LPSM)}\footnote{MPSM can be treated as a special case of LPSM, in which $e_i(r)=1 \forall i\in\mathcal{M}$} is a stable matching, such that each of the paired nodes is paired with node in its preference list limited by individual PEF's of the paired nodes. 

Node $i$ will limit its preference to top most $\max\left(\lfloor e_i(r)\left|l_i(r)\right|\rfloor,1\right)$ nodes in its preference list, where $e_i(r)$ is Preference Exploration Factor in $r^{th}$ slot and $\left|l_i(r)\right|$ denotes the degree of node $i$ in $G(r)$. If values of PEF is very low, $e_i(r)\leq\dfrac{1}{m-1}$, then effectively $G(r)$ will be $G_f(r)$ as nodes will only have top preference in their preference list.

For each unpaired node $i$ in $r^{th}$ slot, $s_i(r)=\{0\}$ and node $i$ will download a segment from $\left(O_i(r)\right)^c$ uniformly at random with Segment Aggressive Probability, $a_i(r)$. 

Each of the node will be follow the \emph{Limited Stable Pairing Algorithm} as outlined in Algorithm \ref{alg:lspa}. If each of the nodes follows Algorithm \ref{alg:lspa}, then LPSM will be obtained in each slot. Algorithm \ref{alg:lspa} is based on algorithms being proposed for various variants of Stable Roommates Problem in \cite{srp1985,srp2002}. These algorithms are extended to incorporate PEF and SAP.

\begin{algorithm}[h]
	\caption{Limited Stable Pairing Algorithm for node $i$}
	\label{alg:lspa}
	\begin{algorithmic}[1]
		\Require Knowledge of $O_j(r)\forall j\in\mathcal{M}$ to node $i$ and $e_i(r), a_i(r) \forall r>0$
		\For {all slots $r>0$}
		\State Arrange all nodes in order of their decreasing incremental utilities.
		\State Choose top $\max(1,\lfloor e_i(r)\left|l_i(r)\right|\rfloor)$ nodes from the list
		\While {(Preference list is non-empty) or (node $i$ does not have a partner)}
		\State Propose a link to node $j$ at head of preference list
		\If {node $j$ accepts $i$}
		\State Activate link between node $i$ and $j$.
		\EndIf
		\EndWhile
		\If {node $i$ is unable to find a pairing node}
		\State Randomly choose $X$ with probability distribution $\mathbb{P}(X=1)=a_i(r)=1-\mathbb{P}(X=0)$
		\If {X=1}
		\State Choose $e\in O_i^c(r)$ uniformly at random.
		\State $O_i(r+1)\gets O_i(r)\bigcup\left\{e\right\}$
		\State $c_i(r)\gets c_i(r-1)+c_{exp}(e)$
		\EndIf
		\EndIf
		\State $r\gets r+1$
		\EndFor
	\end{algorithmic}
\end{algorithm}

Algorithm \ref{alg:lspa} will be reduced to \emph{Preferential Exploration Pairing Algorithm} if nodes in social group decide not to download new segments i.e. $a_i(r)=0\, \forall\, i\in\mathcal{M}, r>0$. However, nodes in each slot limit their preferences to top $\max\left(1,\lfloor e_i(r)\left|l_i(r)\right|\rfloor\right)$ nodes returning the maximum utility. We will obtain LPSM matching in each slot of Algorithm \ref{alg:pepa}.

\begin{algorithm}[h]
	\caption{Preferential Exploration Pairing Algorithm for node $i$}
	\label{alg:pepa}
	\begin{algorithmic}[1]
		\Require Knowledge of $O_j(r)\forall j\in\mathcal{M}$ to node $i$ and $e_i(r) \forall r>0$
		\For {all slots $r>0$}
		\State Arrange all nodes in $l_i(r)$ in decreasing order of their incremental utilities.
		\State Choose top $\max(1,\lfloor e_i(r)\left|l_i(r)\right|\rfloor)$ nodes from the list
		\While {(Preference list is non-empty) or (node $i$ does not have a partner)}
		\State Propose a link to node $j$ at head of preference list
		\If {node $j$ accepts $i$}
		\State Activate link between node $i$ and $j$.
		\EndIf
		\EndWhile
		\State $r\gets r+1$
		\EndFor
	\end{algorithmic}
\end{algorithm}

If nodes are open to pairing with any of the nodes in all slots, i.e. $e_i(r)=1\forall i\in\mathcal{M}, r>0$, then Algorithm \ref{alg:pepa} can be further reduced to \emph{Link for Sure} algorithm. In this case, MPSM matching will be obtained in each slot.
\begin{algorithm}[h]
	\caption{Link for Sure Algorithm for node $i$}
	\label{alg:lfs}
	\begin{algorithmic}[1]
		\Require Knowledge of $O_j(r)\forall j\in\mathcal{M}$ to node $i$
		\For {all slots $r>0$}
		\State Arrange all nodes in $l_i(r)$ in decreasing order of their incremental utilities.
		\While {(List is non-empty) or (node $i$ does not have a partner)}
		\State Propose a link to node $j$ at head of list
		\If {node $j$ accepts $i$}
		\State Activate link between node $i$ and $j$.
		\EndIf
		\EndWhile
		\State $r\gets r+1$
		\EndFor
	\end{algorithmic}
\end{algorithm}

\subsection{Randomized Strategy Selection}
\label{subsec:rand_sel}
For deterministic strategy selection, we assume nodes know about segment sets available with each of the nodes in each slot (Assumption A\ref{asmp:node_knows}), which requires overhead communication among nodes with the help of facilitator. This assumption might not hold true in general. For such scenarios, nodes would choose their partner nodes at random in each slot. We propose and analyze a randomized algorithm being followed by nodes for selection of nodes to activate links.

In the beginning of each slot, each node $i$, chooses a node among $\mathcal{M}\backslash\left\{i\right\}$. If the chosen node also chooses node $i$, then link is activated among them if the \gt is satisfied. This process is continued until no link exists in graph among the nodes $(\mathcal{L}(r)=\emptyset)$.

\begin{algorithm}[h]
	\caption{Decentralized Randomized Algorithm for node $i$}
	\label{alg:rand_algo}
	\begin{algorithmic}[1]
		\For {all decision epochs}
		\State Choose a node $j$ uniformly at random from $\mathcal{M}\backslash\left\{i\right\}$.
		\If {node $j$ chooses node $i$}
			\If {\gt is satisfied by nodes $i$ and $j$}
			\State Activate link between node $i$ and $j$.
			\EndIf
		\EndIf
		\EndFor
	\end{algorithmic}
\end{algorithm}

Algorithm \ref{alg:rand_algo} is based on  algorithm being used in Bittorrent based peer-to-peer networks \cite{bittorrent_spec,cheapfreeriding} and GT-compliant peer-to-peer networks \cite{agg_sadhana}.

Algorithm \ref{alg:rand_algo} is applicable irrespective of the number of the initial segments available with nodes however, for the sake of analysis, we assume node $i$ has chosen $k\in\left\{1,2,3\cdots n-1\right\}$ segments uniformly at random from the universe $\mathcal{N}$.

\begin{theorem}
	\label{thm:rand_limit}
	If initial segments sets ($O_i$'s) are chosen uniformly at random from $\mathcal{N}$, with $\left|O_i\right|=k, 1\leq k\leq n-1, \forall i\in\mathcal{M}$, then the randomized algorithm is asymptotically optimal\footnote{$(a\mod b)$ denotes the remainder when $a$ is divided by $b$, where $a,b\in\mathbb{N}$ and $a\geq b$.} in $m$, i.e. for large $m$, $$\frac{E\left(\sum_{i\in\mathcal{M}} \left|O_i(r_e)\right|\right)}{nm-\left(m\mod 2\right)}\rightarrow 1$$
\end{theorem}
\begin{proof}
	In each decision epoch $r$, various link activations might take place. A link between node $i$ and $j$ at decision epoch $r$ will be activated iff, 
	\begin{inparaenum}[\it (a)]
		\item node $i$ chooses node $j$;
		\item node $j$ choose node $i$; and
		\item \gt is satisfied between nodes $i$ and $j$ at decision epoch $r$.
	\end{inparaenum}
	
	For notational convenience, let $x_{(r,i)}=\left|O_i(r)\right|$ , where $O_i(r)$ denotes the segment set available with node $i$ in beginning of decision slot $r$. Also, let $s_{(r,i)}$ denote the node chosen\footnote{$s_{(r,i)}$ is different from $s_i(r)$ as defined in Section \ref{sec:sysmodel}.} in slot $r$ by node $i$. At start of first decision slot, each node has $x_{(1,i)}=k$ segments, that are uniformly randomly picked from the universe $\mathcal{N}$. Since algorithm is random, we choose node $i$ without the loss of any generality. Let node $i$ choose node $j_r$ in $r^{th}$ decision slot, i.e $s_{(r,i)}=j_r$, where $j_r$ is chosen uniformly at random from $\mathcal{M}\backslash\left\{i\right\}$.
	
	The segment set at node $i$ in beginning of slot $r$ is $O_i(r)=\left\{e_1,e_2\cdots e_{\left|O_i(r)\right|}\right\}$. For every decision slot $r$, and any segment $e\in O_i(r)$, we define random variables $X_e(j_r,r)$:
	$$X_e(j_r,r)=\begin{cases}
	1\quad if\, e\not\in O_i(r)\\
	0\quad otherwise
	\end{cases}$$
	For slot 1, node $i$'s segment set cardinality can increase after exchange, $\left(x_{(2,i)}-x_{(1,i)}\right)$, will equal the number of segments that are available with node $j_1$ but not with node $i$, iff $s_{(r,j_r)}=i$ i.e.
	$$\left(x_{(2,i)}-x_{(1,i)}\right)=\begin{cases}
	\sum_{e\in O_i(r)} X_e\quad \text{if } \begin{matrix}
	s_{(r,i)}=j_r,\,s_{(r,j_r)}=i,\\
	j_r\in S_i(r)
	\end{matrix}\\
	0\quad otherwise 
	\end{cases}$$
	Now,
	\begin{align}
	& E\left(x_{(2,i)}-x_{(1,i)}\right) \nonumber\\
	&= E\left(\sum_{e\in O_i(1)} X_e\left|\begin{matrix}
	s_{(1,i)}=j_1,\\
	s_{(1,j_1)}=i,\\
	j_1\in S_i(1)
	\end{matrix}\right.\right)\mathbb{P}\left(\begin{matrix}
	s_{(1,i)}=j_1,\\
	s_{(1,j_1)}=i,\\
	j_1\in S_i(1)
	\end{matrix}\right)\nonumber\\
	&\stackrel{(a)}{=}\left|O_{j_1}(1)\right|\mathbb{P}\left(X_e=1\right)\mathbb{P}(s_{(1,i)}=j_1)\mathbb{P}(s_{(1,j_1)}=i)\mathbb{P}(j_1\in S_i(1))\nonumber\\
	&=k\cdot\left(1-\dfrac{k}{n}\right)\cdot\dfrac{1}{m-1}\cdot\dfrac{1}{m-1}\cdot\left(1-\dfrac{1}{{n\choose k}}\right)\nonumber\\
	&=\dfrac{k}{\left(m-1\right)^2}\left(1-\dfrac{k}{n}\right)\left(1-\dfrac{1}{{n\choose k}}\right)
	\label{eq:rand_1}
	\end{align}
	where (a) follows from linearity of expectation and independence of events.
	
	Using $E(x_{(1,i)})=k$, rewriting \eqref{eq:rand_1} as,
	\begin{align}
	\label{eq:rand_2}
	E(x_{(2,i)})=E(x_{(1,i)})+\dfrac{E(x_{(1,i)})}{\left(m-1\right)^2}\left(1-\dfrac{E(x_{(1,i)})}{n}\right)\left(1-\dfrac{1}{{n\choose {E(x_{(1,i)})}}}\right)
	\end{align}
	To proceed further in analysis, we assume the segment sets obtained in the beginning of second slot to uniformly distributed with each node having a cardinality of $E(x_{(2,i)})$.
	
	Generalizing \eqref{eq:rand_2} for any $r$ and approximating ${n\choose k'}\approx \dfrac{\Gamma(n+1)}{\Gamma(k'+1)\Gamma(n-k'+1)}$ for non-integer $k'$, we get,
	\begin{align}
	\label{eq:rand_3}
	& E(x_{(r+1,i)})\approx E(x_{(r,i)}) +\nonumber\\ &\dfrac{E(x_{(r,i)})\Gamma\left(E(x_{(r,i)})+1\right)\Gamma\left(n-E(x_{(r,i)})+1\right)\left[n-E(x_{(r,i)})\right]}{n\Gamma(n+1)(m-1)^2}
	\end{align}
	Sequence $E(x_{(r,i)})$ is monotonically non-decreasing and bounded above by $n$; hence sequence converges to $n$ for large enough $m$. Hence, each node gets $n$ segments asymptotically. Thus, asymptotically in $m$, i.e. for large number of nodes aggregate cardinality approaches optimal aggregate cardinality (by use of Corollary \ref{cor:agg_opt}) \textit{i.e.} $nm-\left(m\mod 2\right)$. Therefore, for large $m$,
	$$\frac{E\left(\sum_{i\in\mathcal{M}} \left|O_i(r_e)\right|\right)}{nm-\left(m\mod 2\right)}\rightarrow 1$$
\end{proof}
	\section{Price of Choices}
\label{sec:poc}
Considering social group with non-aggressive nodes (i.e. $a_i(r)=0\forall i\in\mathcal{M}, r>0$), each node has many choices and can be matched in numerous ways at each slot. Since, the utility being earned by a node also depends on choice made by other node. Each choice of the a node in a slot can lead to different matching, hence can lead to different utilities for the node towards the end. Only few nodes will be able to get the desired utilities towards the end of all exchanges. However, there exists a series of choices such that aggregate utility of all nodes can be maximized towards the end of all exchanges. Owing to their selfishness to maximize their individual utilities, each nodes makes choices which might be optimal for it, but not for the social group.

Such situations are observed in various other scenarios as well, where owing to selfish behavior of a agents/players system tends to deviate from the optimal behavior. The Price of Anarchy (PoA) \cite{poa} is one such game theoretic concept to measure degradation in the efficiency of a system due to selfish behavior of agents involved. It is a general notion that has been extended to diverse systems and notions of efficiency \cite{routing_poa,queue_poa,social_poa,game_poa}. The Price of Anarchy is defined as the ratio between the optimal `centralized' solution and the `worst equilibrium'. Different notions of equilibrium lead to variations in notion of Price of Anarchy as Pure PoA \cite{pure_poa,pure_poa_leme}, Strong PoA \cite{spoa}, Bayes-Nash PoA \cite{pure_poa_leme}, Price of Stability \cite{price_stability} etc.

We define \emph{Price of Choices (PoC)\footnote{PoC is only defined for cases where nodes are non-aggressive, i.e. $a_i(r)=0\forall i\in\mathcal{M}, r>0$}} as the ratio between the optimal 'centralized' solution and the 'utilitarian function' at the end.
$$PoC (\hat{S}(r_{end}))=\frac{\alpha^*}{U(r_{end})}$$
where $\alpha^*$ denotes the optimal `centralized' solution and $U(r_{end})$ denotes the aggregate utility of all nodes at the end of all exchanges if choices specified by $\hat{S}(r_{end})$ are being made.

For further analysis, we assume $u_i(r)=\left|O_i(r)\right|$. Hence, $U(r_{end})$ denotes the aggregate cardinality of all nodes at the end of all exchanges. Correspondingly, $\alpha^*$ denotes the optimal aggregate cardinality that can be obtained by the GT compliant exchanges in the social group. The problem of computing $\alpha^*$ has been addressed in \cite{social_optimum}.
	\section{Performance Evaluation}
\label{sec:performance}
\subsection{Limited Stable Pairing Algorithm(LSPA) Evaluation}
We evaluate the performance of the Algorithm \ref{alg:lspa} namely, \emph{Limited Stable Pairing Algorithm (LSPA)} for different social group scenarios. Each scenario consists of \emph{Segment Aggressive Probability (SAP)} chosen from a discrete set, \emph{Preferential Exploration Factor (PEF)} chosen from a discrete set, and a 3-tuple (number of nodes $m$, universe size $n$, initial segment size $k$). For sake of simplicity, in each scenario SAP and PEF are same across all nodes and slots. For each scenario we perform Monte Carlo simulations where the initial segments sets are chosen uniformly at random. 

We consider the two metrics of interest namely, 
\begin{itemize}
	\item \emph{Normalized Mean Aggregate Cardinality (NMAC):} It is defined as the ratio of Mean Aggregate cardinality to $mn$. NMAC characterizes the fraction of segments which have been obtained by all nodes either by using the inter-node links or downloading from the server for a given value of SAP and PEF. Mathematically, NMAC is the mean of $\frac{\sum_{i\in\mathcal{M}}\left|O_i\left(r_{end}\right)\right|}{mn}$ over all trials.
	\item \emph{Normalized Mean number of Segments Downloaded (NMSD):}It is defined as the ratio Mean number of segments downloaded from server to $mn$. NMSD characterizes the fraction of segments that has been obtained from server. Mathematically, NMSD is the mean of $\frac{\sum_{i\in\mathcal{M}}c_i \left(r_{end}\right)}{mn}$
\end{itemize}
\begin{figure*}
	\centering
	\subfigure[Normalized Mean Aggregate Cardinality]{
		\label{fig:agg_20_50_6}
		\includegraphics[width=0.45\textwidth]{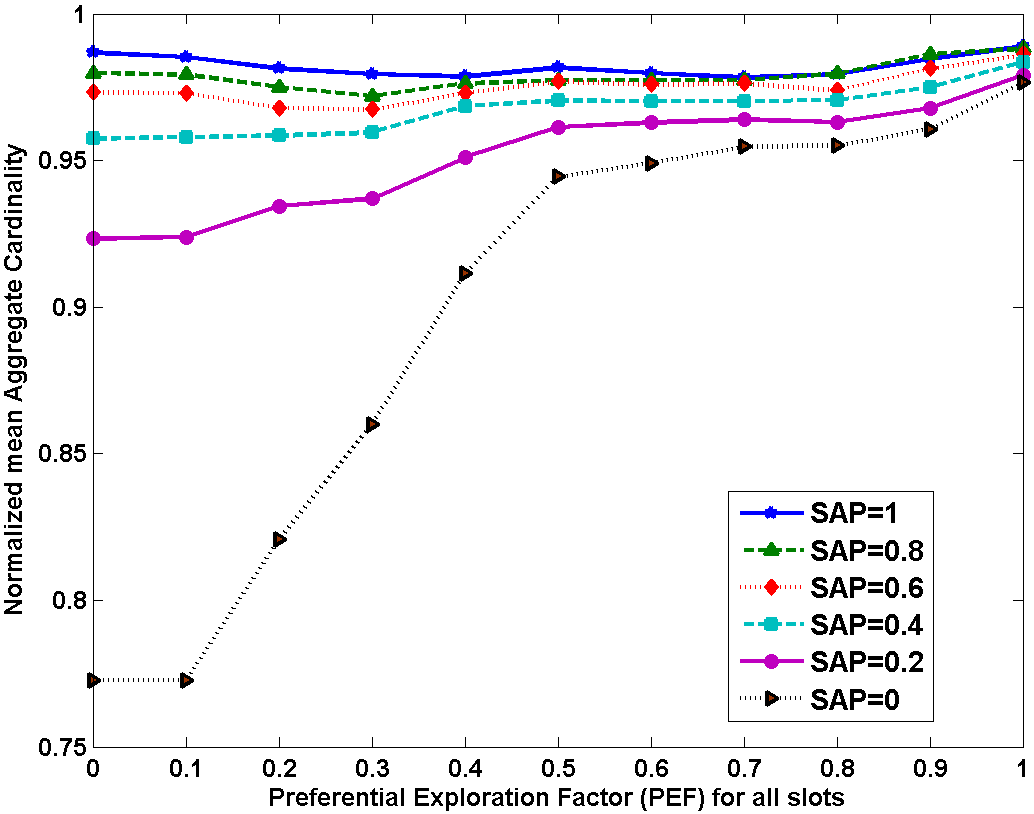}}
	\subfigure[Normalized Mean number of Segments Downloaded]{
		\label{fig:cost_20_50_6}
		\includegraphics[width=0.45\textwidth]{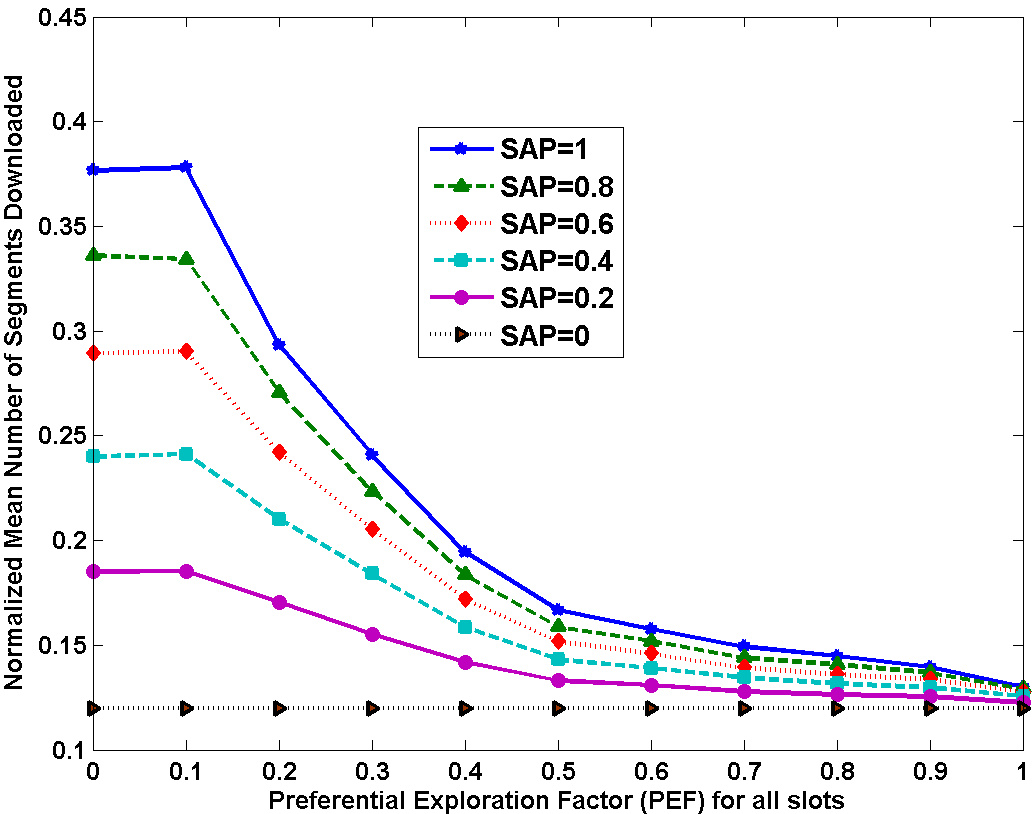}}
	\caption{Variation of Normalized Mean Aggregate Cardinality and Normalized Mean number of Segments Downloaded with PEF for different values of SAP for $\left(m,n,k\right)\equiv\left(20,50,6\right)$}
	\label{fig:lspa_20_50_6}
\end{figure*}
\begin{figure*}
	\centering
	\subfigure[Normalized Mean Aggregate Cardinality]{
		\label{fig:agg_30_60_5}
		\includegraphics[width=0.45\textwidth]{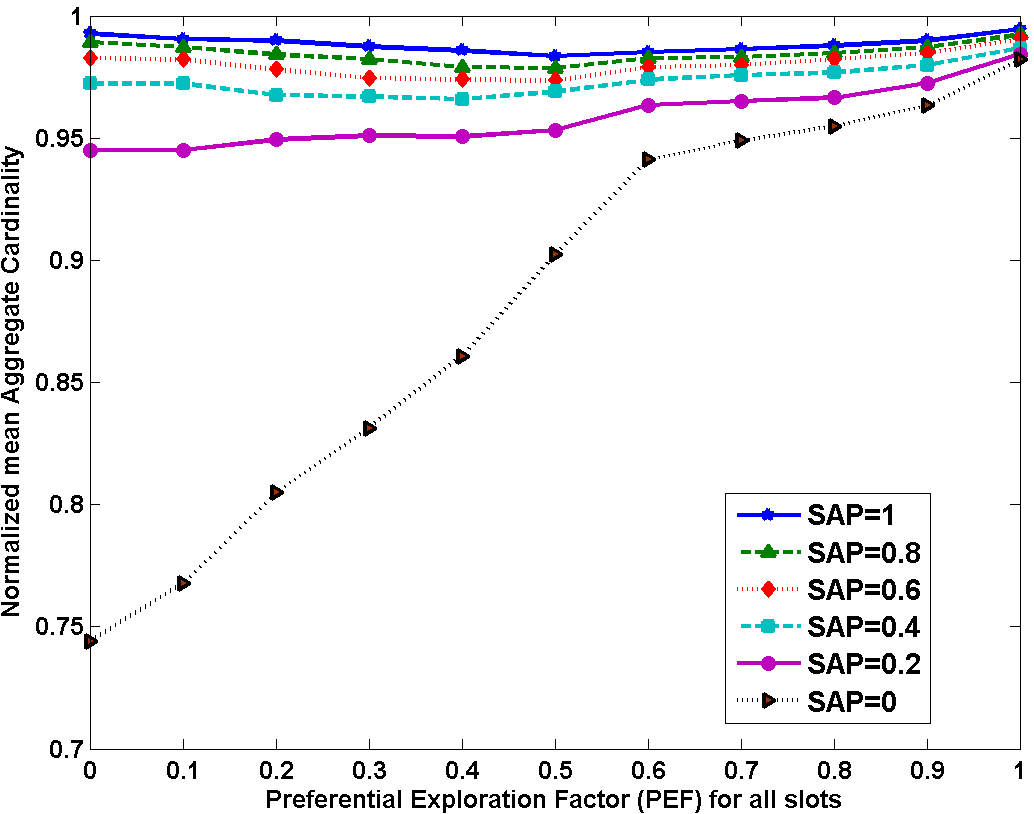}}
	\subfigure[Normalized Mean number of Segments Downloaded]{
		\label{fig:cost_30_60_5}
		\includegraphics[width=0.45\textwidth]{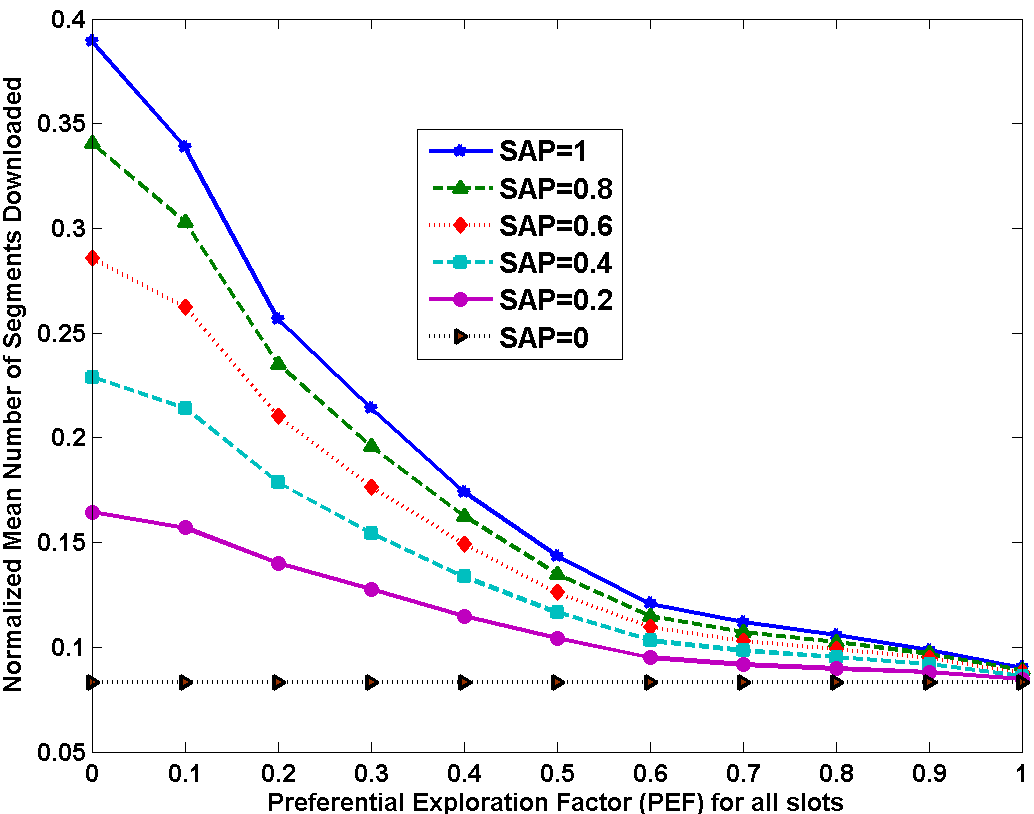}}
	\caption{Variation of Normalized Mean Aggregate Cardinality and Normalized Mean number of Segments Downloaded with PEF for different values of SAP for $\left(m,n,k\right)\equiv\left(30,60,5\right)$}
	\label{fig:lspa_30_60_5}
\end{figure*}
Figures \ref{fig:lspa_20_50_6} and \ref{fig:lspa_30_60_5} shows the variation of NMAC and NMSD with PEF for different values of SAP for social groups consisting of varied sets of $(m,n,k)$ For each set of $(m,n,k)$, $500$ sample points were used to compute NMAC and NMSD. For a given SAP, with the increasing PEF fraction of universe obtained by node increases. but number of segments downloaded from the server decreases. However, for a given PEF with the increasing SAP, fraction of universe obtained by nodes also increases at the expense of increased number of segments being downloaded from the server. Therefore, low value of SAP (i.e. lower aggression for new segments) and high PEF (i.e. exploring large number of possibilities for link formation) can yield high fraction of universe being available to nodes (i.e. high utility) towards the end of system. Ideally, all nodes in Social Groups should have $SAP=0$ and $PEF=1$.
\subsection{Price of Choices}
\begin{figure*}[t]
	\centering
	\includegraphics[width=0.9\textwidth]{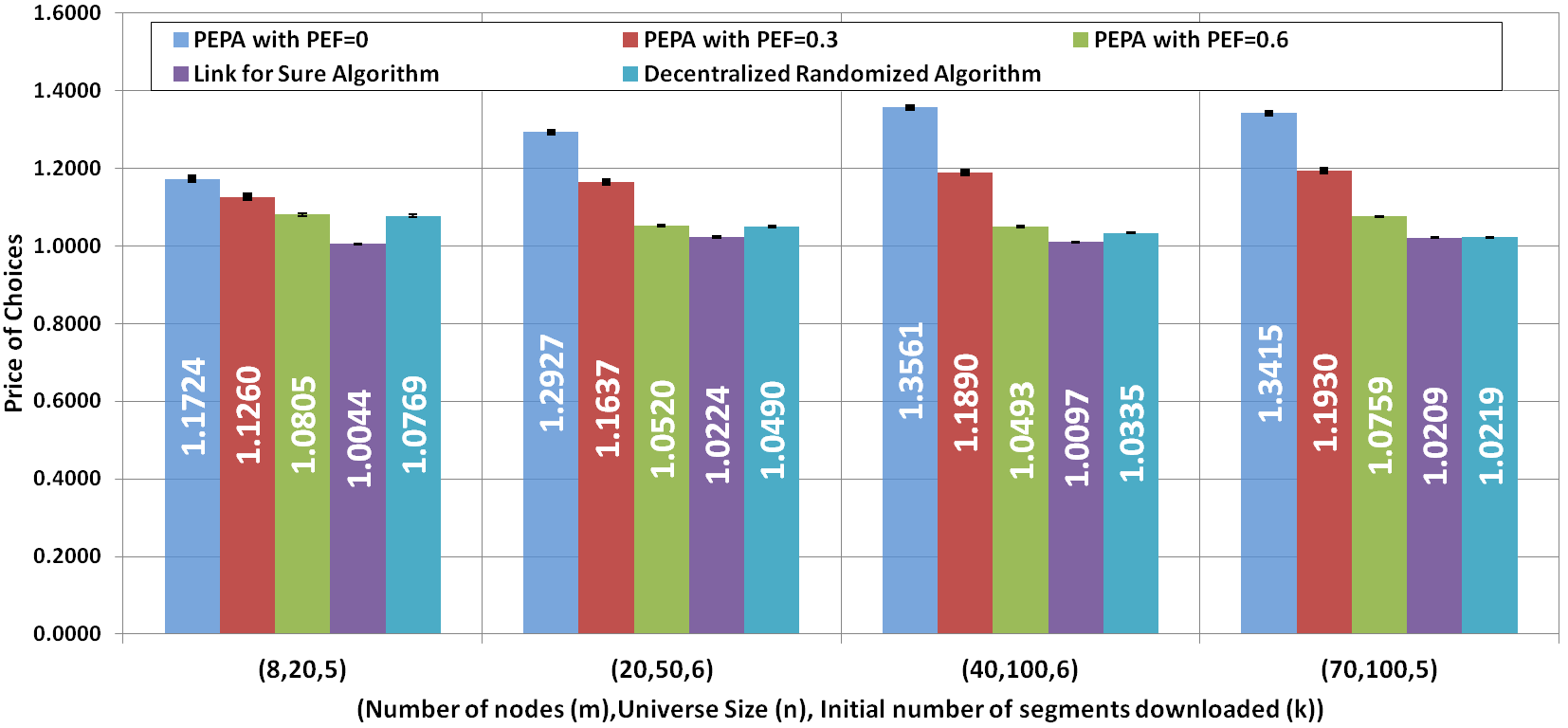}
	\caption{Expected value of Price of choices achieved by PEPA for different values of PEF, Link for Sure Algorithm and Decentralized Randomized Algorithm, averaged over 1000 randomly generated segment sets for 4 different problem instances (m,n,k).}
	\label{fig:multinode_results}
\end{figure*}
Figure \ref{fig:multinode_results} evaluates the performance of various algorithms (Algorithm \ref{alg:pepa},\ref{alg:lfs} and \ref{alg:rand_algo}) namely, \emph{Preferential Exploration Pairing Algorithm (PEPA)} for various values of PEF, \emph{Link for Sure Algorithm (LSA)} and \emph{Randomize Algorithm}. For a particular set of (number of nodes $m$, universe size $n$, initial segment size $k$) values, we generate $1000$ segment sets uniformly at random. For each such `sample point' or 'run' the optimal aggregate cardinality is computed using the algorithms/methos described in \cite{social_optimum} and each of the algorithms simulated. For each sample run, PoC is computed for each of the algorithm. Mean PoC (over $1000$ sample points) are shown, along with $95\%$ confidence interval.

Our results show that Link for Sure Algorithm (LSA) performs the best-- it is able to achieve the lowest Price of Choices (PoC) close to 1. The value of PoC increases, as PEF decreases for PEPA. PoC for Randomized Algorithm also closely follows PoC for LSA. In spite of differences among LSA and Randomized algorithm, one can observe that values of PoC obtained for the simulations is within $10\%$ of the desired value of PoC, i.e. 1.

Some key observations:
\begin{itemize}
	\item PoC for Algorithm \ref{alg:lfs} LSA is very close to 1.
	\item PoC for Algorithm \ref{alg:pepa} decreases with the increasing value of Preferential Exploration Factor (PEF).
	\item \emph{PoC for Randomized algorithm decreases as the number of nodes increases:} This follows from Theorem \ref{thm:rand_limit}. Randomized algorithm performs close to optimal as number of nodes increases, which leads to decrease in the value of PoC as verified using simulations.
\end{itemize}

Our results indicate that following decentralized algorithm \emph{Link for Sure} Algorithm will bring it quite close to optimal scenarios. Also, the linear complexity of $O(m)$ for the this algorithm\footnote{Complexity for computing \emph{Maximal Partial Stable Matching (MPSM)} in each slot is same as that of computing Stable roommate matching \textit{i.e.} $O(m^2)$} makes them quite useful for practical applications as well. However, social group needs to bear cost for overhead for sending information about segment sets to all nodes (Assumption A\ref{asmp:node_knows}).

In case, social group does not want to bear the cost of overhead for sending information, nodes can use the Randomized algorithm at the expense of inefficiency as compared to Link for Sure Algorithm.
	\section{Conclusions}
\label{sec:conclusions}
The problem studied in this paper is motivated in context of a common network architecture namely Social Groups (being observed in various socio-technological networks namely, Cellular Network assisted Device-to-Device, Cloud assisted Peer-to-Peer Networks, hybrid Peer-to-Peer Content Distribution Network and Direct Connect Networks), where each member is interested in maximizing its own utility (which is an increasing function of cardinality of node). These decentralized mutual exchanges will not only get nodes segments at low cost, but also decrease the cost of operation for the central server. However, to tackle the problem of non-reciprocating behavior arising in such situations, we use the GT criterion \cite{social_optimum}, and explore a number of algorithms for exchange of segments, with each exchange to be GT-compliant. 

Nodes based on their aggressiveness for new segments and preferential exploration nature can choose its strategies using \emph{Limited Stable Pairing Algorithm}. If node is not aggressive for new segments, but still restricted about its preferences for exchanges, then \emph{Preferential Exploration Pairing Algorithm} can be used for choosing the strategy in each slot. \emph{Link for Sure Algorithm} can be used by nodes for deciding their strategies, which do not restrict their preferences. Following (practically implementable linear complexity) LFS algorithm for deciding the myopic strategy in each slot, we observe that PoC is within $3\%$ of the desired PoC of 1 at no additional cost to nodes. This also reduces the need to consider the impact of choosing a strategy in a slot on future strategies as well which might come at a high computational cost to node.

However, all algorithms require information about the availability of segment sets in each slot, which might be unavailable, in that case nodes can use randomized algorithm. 

To summarize, \emph{low aggressiveness for new segments and exploration of maximum possibilities for link formation is in interest of the node as well as Social Group}.
\section{Future Work}
\label{sec:future}
Some of the algorithms in this paper require information about availability of segment sets in each slot to be made available to each node. The same can be practically implemented with the help of facilitator (which plays role similar to tracker in P2P networks). Inherently, it was assumed that nodes provide the information on segment availability truthfully, which might not be true. There might be scenarios where falsifying information can give undue benefits to some node(s).

Also, in this work, we have assumed that all nodes have homogeneous utility function, which is a strictly increasing function of cardinality of node. However, node(s) may have utility functions which do not fall within the scope. Hence, we will like to study the node's evolution when social group consists of nodes with heterogeneous and generic utility functions. With heterogeneous utility functions, nodes might choose to follow different algorithms. We will like to analyze scenarios where nodes have heterogeneous and generic utility functions along with following different decentralized algorithms. Further, we will like to consider systems with different exchange criterion.
	\bibliography{info_exchange}
	\appendices
\section{Stable Roommates Problem and Variants}
The stable-roommate problem (SRP) is the problem of finding a stable matching among a pair of elements, such that, there is no pair of elements, each from a different matched set, where each member of the pair prefers the other to their match\cite{srp1985}. In a given instance of the stable-roommates problem (SRP), each of $m$ participants ranks the other participants in order of preference. A matching is a set of $\frac{m}{2}$ disjoint pairs of participants. A matching $M$ in an instance of SRP is stable if there are no two participants $x$ and $y$, each of whom prefers the other to his partner in $M$. Such a pair is termed as stable pair. 

Multiple variants of the SRP have been proposed and studied by researchers. We consider one such variant namely, Stable Roommates Problem with Ties and Incomplete lists (SRPTI), in which, the participants are allowed to have ties among participants and can chose to ignore certain participants \cite{srp2002}. Our problem is each slot can be mapped to an instance of SRPTI. Stable Matching as defined in \cite{srp1985,srp2002} exists iff all participants are able to find a pairing participant. But, there might exist some matchings in which only some participants are able to find stable pairing participants. We term such matchings as \emph{Partial Stable Matching}.

We have defined \emph{Limited Preference Stable Matching} and \emph{Maximal Partial Stable Matching} in context of our problem.


\section{}
\begin{lemma}[Rewriting Lemma 1 of \cite{social_optimum}]
	For any order of link activations (resulting in a completely disconnected graph) and $O_i\subsetneq \bigcup_{j\in\mathcal{M}} O_j \forall i\in\mathcal{M}$, at least two nodes will have $\bigcup_{i\in\mathcal{M}} O_i$.
\end{lemma}
\begin{proof}
	Please refer to proof for Lemma 1 of \cite{social_optimum}.
\end{proof}
\begin{corollary}
	\label{cor:agg_opt}
	For any order of link activations (resulting in a completely disconnected graph) optimal aggregate cardinality is upper bounded by $nm-\left(m\mod 2\right)$.
\end{corollary}
\begin{proof}
	Trivial upper bound on aggregate cardinality is given by $nm$ as system consists of $m$ nodes and $n$ segments. For odd number of nodes, it follows from the fact that number of nodes with universe can be even only. Hence, considering the best possible scenario where $m-1$ nodes have got the universe. And node without universe will have at least 1 segment missing, therefore, optimal aggregate cardinality is upper bounded by $nm-\left(m\mod 2\right)$.
\end{proof}
\end{document}